\documentclass[a4paper,UKenglish,cleveref, autoref, thm-restate]{lipics-v2021}
\pdfoutput=1
\hideLIPIcs  
\usepackage{graphicx}

\usepackage[T1]{fontenc}
\usepackage{amsmath}
\usepackage{thmtools}
\usepackage{thm-restate}
\usepackage{enumerate}
\usepackage{ascmac}
\usepackage{comment}
\usepackage{color}
\usepackage{thm-restate}
\usepackage{algorithm}
\usepackage{algpseudocode}

\newtheorem{fact}[theorem]{Fact}

\newcommand{\occ}{\mathit{occ}}

\newcommand{\NOcc}{\mathsf{NO}}
\newcommand{\ENO}{\mathsf{ENO}}
\newcommand{\NUS}{\mathsf{NUS}}
\newcommand{\sqsuf}{\mathit{sqSuf}}
\newcommand{\lrsuf}{\mathit{lrSuf}}
\newcommand{\sqpref}{\mathit{sqPref}}
\newcommand{\lrpref}{\mathit{lrPref}}

\newcommand{\LB}{\mathsf{LB}}
\newcommand{\RB}{\mathsf{RB}}
\newcommand{\LM}{\mathsf{LM}}
\newcommand{\RM}{\mathsf{RM}}
\newcommand{\M}{\mathsf{M}}

\newcommand{\MUS}{\mathsf{MUS}}

\newcommand{\STree}{\mathsf{STree}}
\newcommand{\CDAWG}{\mathsf{CDAWG}}
\newcommand{\e}{\mathsf{e}}
\newcommand{\rb}{\mathsf{r}}

\title{Space-Efficient Online Computation of String Net Occurrences}

\author{Takuya Mieno}{Department of Computer and Network Engineering, University of Electro-Communications, Chofu, Japan}{tmieno@uec.ac.jp}{https://orcid.org/0000-0003-2922-9434}{}
\author{Shunsuke Inenaga}{Department of Informatics, Kyushu University, Fukuoka, Japan}{inenaga.shunsuke.380@m.kyushu-u.ac.jp}{https://orcid.org/0000-0002-1833-010X}{} 
\authorrunning{Mieno~and~Inenaga} 

\Copyright{Takuya Mieno and Shunsuke Inenaga} 

\ccsdesc[500]{Mathematics of computing~Combinatorial algorithms}
\keywords{string net occurrences, suffix trees, CDAWGs, maximal repeats, minimal unique substrings (MUSs)} 
\category{} 
\relatedversion{} 

\nolinenumbers 

\EventEditors{John Q. Open and Joan R. Access}
\EventNoEds{2}
\EventLongTitle{42nd Conference on Very Important Topics (CVIT 2016)}
\EventShortTitle{CVIT 2016}
\EventAcronym{CVIT}
\EventYear{2016}
\EventDate{December 24--27, 2016}
\EventLocation{Little Whinging, United Kingdom}
\EventLogo{}
\SeriesVolume{42}
\ArticleNo{23}
\begin{document}

\maketitle

\begin{abstract}
  A substring $u$ of a string $T$ is said to be a \emph{repeat} if $u$ occurs at least twice in $T$.
  An occurrence $[i..j]$ of a repeat $u$ in $T$ is said to be a \emph{net occurrence} if
  each of the substrings $aub = T[i-1..j+1]$, $au = T[i-1..j+1]$, and $ub = T[i..j+1]$ occurs exactly once in $T$.
  The occurrence $[i-1..j+1]$ of $aub$ is said to be an \emph{extended} net occurrence of $u$.
  Let $T$ be an input string of length $n$ over an alphabet of size $\sigma$,
  and let $\ENO(T)$ denote the set of extended net occurrences of repeats in $T$.
  Guo et al. [SPIRE 2024] presented an online algorithm which can report $\ENO(T[1..i])$ 
  in $T[1..i]$ in $O(n\sigma^2)$ time, for each prefix $T[1..i]$ of $T$.
  Very recently, Inenaga [arXiv 2024] gave a faster online algorithm that can report $\ENO(T[1..i])$ in optimal $O(\#\ENO(T[1..i]))$ time for each prefix $T[1..i]$ of $T$, where $\#S$ denotes the cardinality of a set $S$.
  Both of the aforementioned data structures can be maintained in $O(n \log \sigma)$ time and occupy $O(n)$ space, where the $O(n)$-space requirement comes from the suffix tree data structure. In particular, Inenaga's recent algorithm is based on Weiner's right-to-left online suffix tree construction.
  In this paper, we show that one can modify Ukkonen's left-to-right online suffix tree construction algorithm in $O(n)$ space, so that $\ENO(T[1..i])$ can be reported in optimal $O(\#\ENO(T[1..i]))$ time for each prefix $T[1..i]$ of $T$.
  This is an improvement over Guo et al.'s method that is also based on Ukkonen's algorithm.
  Further, this leads us to the two following space-efficient alternatives:
  \begin{itemize}
    \item A \emph{sliding-window} algorithm of $O(d)$ working space that can report 
      $\ENO(T[i-d+1..i])$ in optimal $O(\#\ENO(T[i-d+1..i]))$ time for each sliding window $T[i-d+1..i]$ of size $d$ in $T$.
    \item A \emph{CDAWG}-based online algorithm of $O(\e)$ working space that can report 
      $\ENO(T[1..i])$ in optimal $O(\#\ENO(T[1..i]))$ time for each prefix $T[1..i]$ of $T$, where $\e < 2n$ is the number of edges in the CDAWG for $T$.
  \end{itemize}
  All of our proposed data structures can be maintained in $O(n \log \sigma)$ time
  for the input online string $T$.
  We also discuss that the extended net occurrences of repeats in $T$ can be fully characterized 
  in terms of the \emph{minimal unique substrings} (\emph{MUSs}) in $T$.
\end{abstract}

\section{Introduction}\label{sec:intro}

Finding \emph{repeats} in a string is a fundamental task of string processing that has applications
in various fields including bioinformatics, data compression, and natural language processing.
This paper focuses on the notion of \emph{net occurrences} of a repeat in a string, which has attracted recent attention.
Let $u$ be a repeat in a string $T$ such that $u$ occurs at least twice in $T$.
An occurrence $[i..j]$ of a repeat $u$ in $T$ is said to be a net occurrence of $u$ if 
extending the occurrence to the left or to the right results in a unique occurrence,
i.e., each of $aub = T[i-1..j+1]$, $au = T[i-1..j]$, and $ub = T[i..j+1]$ occurs exactly once in $T$.
Finding string net occurrences are motivated for Chinese language text processing ~\cite{lin2001extracting,lin2004properties}.
The occurrence $[i-1..j+1]$ of $aub$ is said to be an \emph{extended net occurrence} of a repeat $u$ in $T$,
and let $\ENO(T)$ denote the set of all extended net occurrences of repeats in $T$.

Guo et al.~\cite{GuoCPM2024} were the first who considered the problem of computing (extended) net occurrences of repeats in a string
from view points of string combinatorics and algorithmics.
Guo et al.~\cite{GuoCPM2024} showed a necessary and sufficient condition for a net occurrence of a repeat,
which, since then, has played a core role in efficient computation of string net occurrences.
For an input string $T$ of length $n$,
they gave an \emph{offline} algorithm for computing $\ENO(T)$ in $O(n)$ time and space for integer alphabets of size polynomial in $n$,
and in $O(n \log \sigma)$ time and $O(n)$ space for general ordered alphabets of size $\sigma$.
Their offline method is based on the suffix array~\cite{ManberM93} and the Burrows-Wheeler transform~\cite{BWT94}.
Ohlebusch et al. gave another offline algorithm that works fast in practice~\cite{OhlebuschBO24}.

Later, Guo et al.~\cite{GuoUWZ24} proposed an \emph{online} algorithm for computing all string net occurrences of repeats.
Their algorithm maintains a data structure of $O(n)$ space that reports $\ENO(T[1..i])$
in $O(n\sigma^2)$ time for each prefix $T[1..i]$ of an online input string $T$ of length $n$.
Since their algorithm computes all the net occurrences upon a query,
their algorithm requires at least $O(n^2\sigma^2)$ time to \emph{maintain
and update} the list of all (extended) net occurrences of repeats in an online string.
Their algorithm is based on Ukkonen's left-to-right online suffix tree construction~\cite{Ukkonen95},
that is enhanced with the suffix-extension data structure of Breslauer and Italiano~\cite{BreslauerI12}.

Very recently, Inenaga~\cite{Inenaga2024-arxiv} proposed a faster algorithm that can maintain 
$\ENO(T[1..i])$ for an online string $T[1..i]$ with growing $i = 1, ..., n$ in a total of $O(n \log \sigma)$ time
and $O(n)$ space.
Namely, this algorithm uses only amortized $O(\log \sigma)$ time to update $\ENO(T[1..i])$ to $\ENO(T[1..i+1])$.
The proposed algorithm is based on Weiner's right-to-left online suffix tree construction~\cite{Weiner73}
that is applied to the reversed input string,
and can report all extended net occurrences of repeats in $T[1..i]$ in optimal $O(\#\ENO(T[1..i]))$ for each $1 \leq i \leq n$,
where $\#S$ denotes the cardinality of a set $S$.

In this paper, we first show that 
Ukkonen's left-to-right online suffix tree construction algorithm 
can also be modified so that it can maintain and update $\ENO(T[1..i])$ in a total of $O(n \log \sigma)$ time
with $O(n)$ space, and can report $\ENO(T[1..i])$ in optimal $O(\#\ENO(T[1..i]))$ time for each $i = 1, \ldots, n$.
While this complexity of our Ukkonen-based method is the same as the previous Weiner-based method~\cite{Inenaga2024-arxiv},
our method enjoys the following merits:
\begin{enumerate}
  \item[(1)] Our result shows that the arguably complicated suffix-extension data structure of Breslauer and Italiano is not necessary for online computation of string net occurrences with Ukkonen's algorithm.
  \item[(2)] The new method can be extended to the \emph{sliding suffix trees}~\cite{Larsson96,Senft05,LeonardIBM-arxiv} and the \emph{compact directed acyclic word graphs} (\emph{CDAWGs})~\cite{BlumerBHME87,InenagaHSTAMP05}.
\end{enumerate}
The first point is a simplification and improvement over Guo et al.'s method~\cite{GuoUWZ24} based on Ukkonen's construction.
The second point leads us to the following space-efficient alternatives:
\begin{itemize}
  \item A sliding-window algorithm of $O(d)$ working space that can be maintained in $O(n \log \sigma)$ time
    and can report $\ENO(T[i-d+1..i])$ in optimal $O(\#\ENO(T[i-d+1..i]))$ time for each sliding window $T[i-d+1..i]$ of size $d$ in $T$.
  \item A CDAWG-based online algorithm of $O(\e)$ working space that can be maintained in $O(n \log \sigma)$ time 
    and can report $\ENO(T[1..i])$ in optimal $O(\#\ENO(T[1..i]))$ time for each prefix $T[1..i]$ of $T$, where $\e$ is the number of edges in the CDAWG for $T$.
\end{itemize}
We note that $\e < 2n$ always holds~\cite{BlumerBHME87}, 
and $\e$ can be as small as $O(\log n)$ for some highly repetitive strings~\cite{Rytter06,RadoszewskiR12}.
Finally, we also discuss that the extended net occurrences of repeats in a string $T$ can be fully characterized 
with the \emph{minimal unique substrings} (\emph{MUSs})~\cite{IlieS11} in $T$. \section{Preliminaries}\label{sec:pre}
\paragraph*{Strings.}
Let $\Sigma$ be an alphabet.
An element of $\Sigma$ is called a character.
An element of $\Sigma^\star$ is called a string.
The empty string $\varepsilon$ is the string of length $0$.
If $T = pfs$ holds for strings $T, p, f$, and $s$,
then $p, f$, and $s$ are called a prefix of $T$, a substring of $T$, and a suffix of $T$, respectively.
A prefix $p$ (resp. a suffix $s$) of $T$
is called a \emph{proper} prefix (resp. a \emph{proper} suffix) of $T$ if $p \neq T$ (resp. $s \neq T$).
For a string $T$, $|T|$ denotes the length of $T$.
For a string $T$ and an integer $i$ with $1\le i \le |T|$,
$T[i]$ denotes the $i$th character of $T$.
For a string $T$ and integers $i, j$ with $1\le i \le j \le |T|$,
$T[i.. j]$ denotes the substring of $T$ starting at position $i$ and ending at position $j$.
For strings $T$ and $w$, we say $w$ \emph{occurs} in $T$ if $T[i.. j] = w$ holds for some $i,j$.
Also, if such $i,j$ exist, we denote by $[i.. j]$ the \emph{occurrence} of $w = T[i.. j]$ in $T$.
Also, we denote by $\occ_T(w)$ the set of occurrences of $w$ in $T$,
i.e., $\occ_T(w) = \{[i.. j] \mid T[i.. j] = w\}$.
For any set $S$, we denote by $\#S$ the cardinally of $S$.
For convenience, 
we assume that the empty string $\varepsilon$ occurs $|T|-1$ times
at the boundaries of consecutive characters in $T$,
and denote these inner occurrences by $T[i+1..i] = \varepsilon$ for $1 \le i < n$.
We also assume that $\varepsilon$ occurs before the first character and after the last character of $T$.
Thus we have $\#\occ_T(\varepsilon) = |T|+1$.
A string $w$ is said to be \emph{unique} in $T$ if $\#\occ_T(w) = 1$.
Also, $w$ is said to be \emph{quasi-unique} in $T$ if $1 \le \#\occ_T(w) \le 2$.
Further, $w$ is said to be \emph{repeating} in $T$ if $\#\occ_T(w) \ge 2$.
We denote by $\lrsuf(T)$~(resp.~$\lrpref(T)$) the longest repeating suffix~(resp.~prefix) of $T$.
We denote by $\sqsuf(T)$~(resp.~$\sqpref(T)$) the shortest quasi-unique suffix~(resp.~prefix) of $T$.

\paragraph*{Maximal repeats and minimal unique substrings.}

A repeating substring of a string $T$ is also called a \emph{repeat} in $T$.
A repeat $u$ in $T$ is said to be 
\emph{left-branching} in $T$ if there are at least two distinct characters $a,a' \in \Sigma$
such that $\#\occ_T(au) \geq 1$ and $\#\occ_T(a'u) \geq 1$.
Symmetrically, a repeat $u$ is said to be \emph{right-branching} in $T$ 
if there are at least two distinct characters
$b,b' \in \Sigma$ such that $\#\occ_T(ub) \geq 1$ and $\#\occ_T(ub') \geq 1$.
A repeat $u$ of $T$ is said to be \emph{left-maximal} in $T$ 
if $u$ is a left-branching repeat of $T$ or $u$ is a prefix of $T$,
and $u$ is said to be \emph{right-maximal} in $T$ if $u$ is a right-branching repeat in $T$ 
or $u$ is a suffix of $T$.
A repeat $u$ of $T$ is said to be \emph{maximal} in $T$ if $u$ is both a left-maximal repeat and a right-maximal repeat in $T$.
Let $\LB(T)$, $\RB(T)$, $\LM(T)$, $\RM(T)$, and $\M(T)$ denote 
the sets of left-branching, right-branching, left-maximal, right-maximal, and maximal repeats in $T$,
respectively.
Note that $\LB(T) \subseteq \LM(T)$, $\RB(T) \subseteq \RM(T)$, and $\M(T) = \LM(T) \cap \RM(T)$ hold.
A unique substring $u = T[i..j]$ of a string $T$ is said to be a \emph{minimal unique substring} (\emph{MUS}) of $T$
if each of the substrings $T[i-1..j]$ and $T[i..j-1]$ is a repeat in $T$.
By definition, no MUS can be completely contained in another MUS in the string,
and thus, there are at most $n$ MUSs in any string $T$ of length $n$.
Let $\MUS(T)$ denote the set of occurrences of all MUSs in $T$.

In what follows, we fix a string $T$ of length $n > 2$ arbitrarily. 

\paragraph*{(Extended) net occurrences.}
For a repeating substring $P$ of $T$, its occurrence $[i.. j]$ in $T$ is said to be a \emph{net occurrence} of $P$ if $T[i.. j] = P$, $T[i-1.. j]$ is unique in $T$, and $T[i.. j+1]$ is unique in $T$.
Let $\NOcc(T)$ be the set of net occurrences in $T$.
We call $T[i-1.. j+1]$ a net unique substring (NUS) if $[i..  j] \in \NOcc(T)$.
Let $\NUS(T)$ be the set of net unique substrings in $T$.
Then, we call the occurrence $[i-1.. j+1]$ of NUS $T[i-1.. j+1]$ the \emph{extended net occurrence} of the repeat $T[i.. j]$.
Let $\ENO(T)$ be the set of extended net occurrences in $T$.
Clearly, there is a one-to-one correspondence between $\NUS(T)$ and $\ENO(T)$, i.e., $T[p.. q] \in \NUS(T)$ iff $[p.. q] \in \ENO(T)$.
Note that $\#\ENO(T) = \#\NOcc(T)$ holds.

\paragraph*{Data structures.}
The \emph{suffix tree}~\cite{Weiner73} of a string $T$ is a compacted trie that represents all suffixes of $T$.
More formally, the suffix tree of $T$ is a rooted tree such that
(1) each edge is labeled by a non-empty substring of $T$,
(2) the labels of the out-edges of the same node begin with distinct characters, and
(3) every suffix of $T$ is represented by a path from the root.
If the path from the root to a node $v$ spells out the substring $w$ of $T$,
then we say that the node $v$ \emph{represents} $w$.
By representing each edge label $x$ with a pair $(i,j)$ of positions
such that $x = T[i..j]$, the suffix tree can be stored in $O(n)$ space.
For convenience, we identify each node with the string that the node represents.
If $av$ is a node of the suffix tree with $a \in \Sigma$ and $v \in \Sigma^\star$,
then the \emph{suffix link} of the node $av$ points to the node $v$.

There are two versions of suffix trees, \emph{implicit suffix trees}~\cite{Ukkonen95} (a.k.a. \emph{Ukkonen trees}) 
and \emph{explicit suffix trees}~\cite{Weiner73} (a.k.a \emph{Weiner trees}).
In the implicit suffix tree of string $T$,
each repeating suffix $s$ of $T$ that has a unique right-extension $c \in \Sigma$
(namely, $\#\occ_{T}(sc) \geq 1$ and $\#\occ_T(sa) = 0$ for any $a \in \Sigma \setminus \{c\}$)
is represented on an edge.
On the other hand, each such repeating suffix is represented by a non-branching internal node
in the explicit suffix tree of $T$.
Let $\STree'(T)$ and $\STree(T)$ denote the implicit suffix tree and the explicit suffix tree of string $T$, respectively.
The internal nodes of $\STree'(w)$ represent the right-branching repeats of $T$,
while the internal nodes of $\STree(w)$ represent the right-maximal repeats of $T$.
It is thus clear that $\STree'(T\$) = \STree(T\$)$ with a unique end-marker $\$$ that does not occur in $T$.
Due to the nature of left-to-right online string processing,
we will use the implicit suffix trees in our algorithms.

\begin{figure}[bht]
  \includegraphics[scale=0.35]{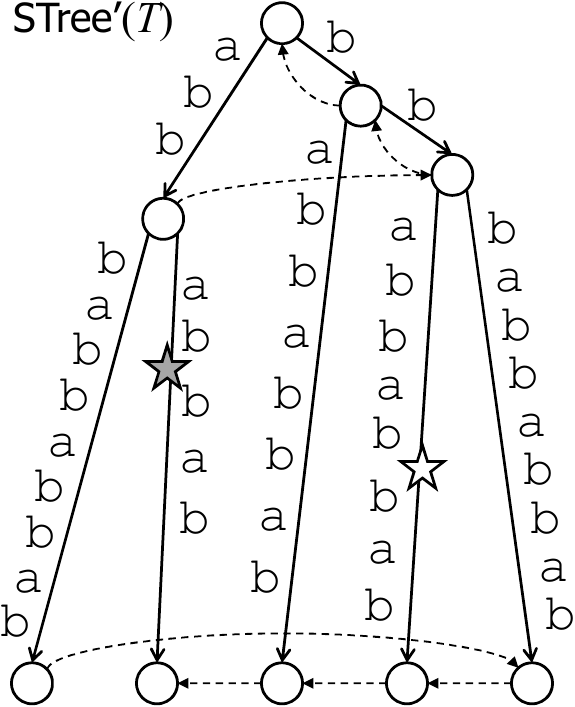}
  \hfill
  \includegraphics[scale=0.35]{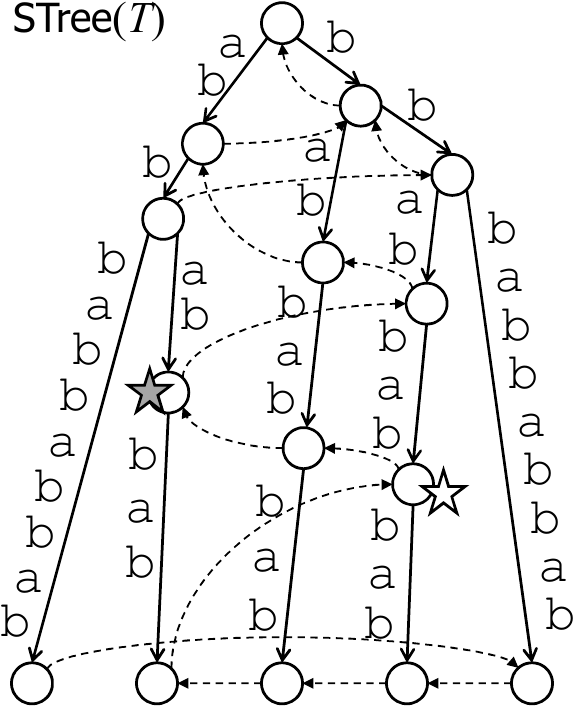}
  \hfil
  \includegraphics[scale=0.35]{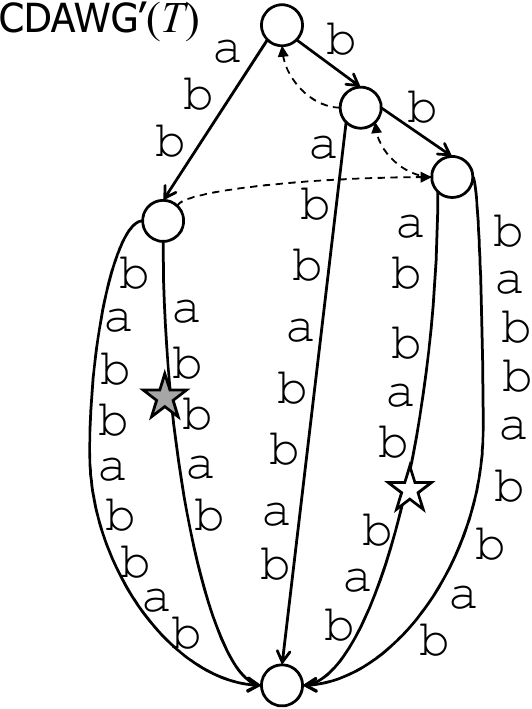}
  \hfil
  \includegraphics[scale=0.35]{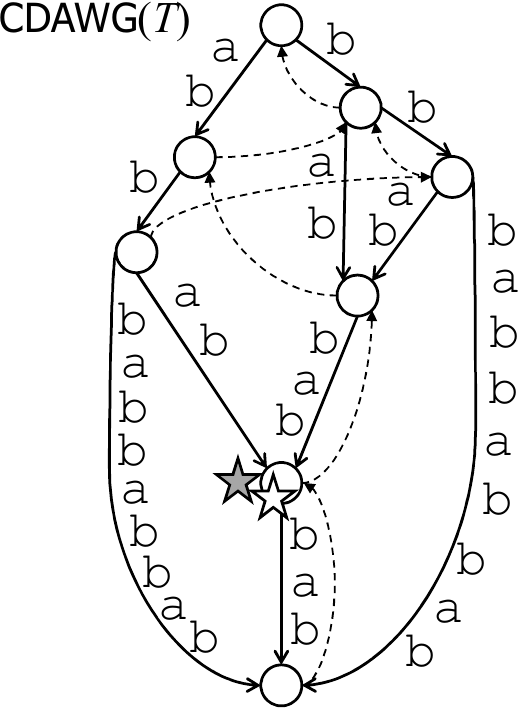}
  \caption{The implicit suffix tree $\STree'(T)$, the explicit suffix tree $\STree(T)$, the implicit CDAWG $\CDAWG'(T)$, and the explicit CDAWG $\CDAWG'(T)$ for string $T = \mathtt{abbbabbabbab}$. The broken arrows represent suffix links. The white and gray stars represent the loci of the longest repeating suffix $\mathtt{bbabbab}$ and shortest quasi-unique suffix $\mathtt{abbab}$ of $T$, respectively.}
  \label{fig:stree}
\end{figure}

The \emph{compact directed acyclic word graph} (\emph{CDAWG})~\cite{BlumerBHME87} of a string $T$ is the smallest 
edge-labeled DAG that represents all substrings of $T$.
There are two versions of CDAWGs as well, \emph{implicit CDAWGs}~\cite{InenagaHSTAMP05} and \emph{explicit CDAWGs}~\cite{BlumerBHME87}.
The implicit CDAWG of string $T$, denoted $\CDAWG'(T)$, is the edge-labeled DAG
that is obtained by merging all isomorphic subtrees of the implicit suffix tree $\STree'(T)$
which are connected by suffix links.
On the other hand, the explicit CDAWG of $T$, denoted $\CDAWG(T)$, is the edge-labeled DAG
that is obtained by merging all isomorphic subtrees of the explicit suffix tree $\STree(T)$
which are connected by suffix links.
The internal nodes of $\CDAWG'(T)$ have a one-to-one correspondence with the left-maximal and right-branching repeats in $T$,
while the internal nodes of $\CDAWG(T)$ have a one-to-one correspondence with the maximal repeats in $T$.
More precisely, the longest string represented by an internal node in $\CDAWG'(T)$
is a left-maximal and right-branching repeat in $T$,
and the longest string represented by an internal node in $\CDAWG(T)$ is a maximal repeat in $T$.
Thus, as in the case of suffix trees, $\CDAWG'(T\$) = \CDAWG(T\$)$ holds with a unique end-marker $\$$.

The length of a path in the implicit/explicit CDAWG is the total length of the labels of the edges in the path.
An in-edge of a node $v$ in the implicit/explicit CDAWG is said to be a \emph{primary edge}
if it belongs to the longest path from the source to $v$.

Due to the nature of left-to-right online string processing,
we will use the implicit CDAWGs in our algorithm.
Let $\e'(T)$ and $\e(T)$ denote the number of edges of $\CDAWG'(T)$ and $\CDAWG(T)$, respectively.
The following lemma guarantees that the worst-case space complexity of our implicit-CDAWG based algorithm is linear in the size of explicit CDAWGs:
\begin{lemma} \label{lem:implicit_CDAWG_size}
  For any string $T$, $\e'(T) \leq \e(T)$.
\end{lemma}
\begin{proof}
  Let $V'$ and $V$ be the sets of nodes of $\CDAWG'(T)$ and $\CDAWG(T)$, respectively.
  We identify each node of $V'$ and of $V$ with the longest string that the node represents.
  Then, $V' = \LM(T) \cap \RB(T)$ and $V = \LM(T) \cap \RM(T) = \M(T)$.
  Since $\RB(T) \subseteq \RM(T)$, we have that $V' \subseteq V$.
  Let $w \in \LM(T) \cap \RB(T)$, which implies that $w \in V'$ and $w \in V$.
  Let $v'$ and $v$ be the nodes that represent $w$ in $\CDAWG'(T)$ and $\CDAWG(T)$, respectively.
  The out-degree $d(v')$ of node $v'$ in $\CDAWG'(T)$,
  as well as the out-degree $d(v)$ of $v$ in $\CDAWG(T)$, 
  are equal to the number of right-extensions $c \in \Sigma$ such that $\#\occ_T(wc) \geq 1$.
  Thus $\e'(T) = \sum_{v' \in V'}d(v') \leq \sum_{u \in V}d(u) = \e(T)$.
\end{proof}
\section{Changes in net occurrences for online string}\label{sec:change}
In this section,
we show how the net unique substrings, equivalently the extended net occurrences in a string $T$ can change when a character is appended.
We will implicitly use the following fact throughout this section.
\begin{fact}\label{obs:occ}
  For any strings $T, w$, and character $c$,
  $\#\occ_{T}(w) \le \#\occ_{Tc}(w)$ holds.
\end{fact}
Thus, if $w$ is repeating in $T$, then it must be repeating in $Tc$.
Further, if $w$ is unique in $Tc$ and is not a suffix of $Tc$, then $w$ must be unique in $T$.

Next lemma characterizes net unique substrings
to be deleted when a character $c$ is appended to string $T$.

\begin{lemma}\label{lem:deletedNOcc}
  Suppose that $aub \in \NUS(T)\setminus\NUS(Tc)$ be a net unique substring in $T$
  where $a, b, c \in \Sigma$ and $u \in \Sigma^\star$.
  Then,
  (i) $au = \sqsuf(Tc)$ and $\#\occ_{Tc}(au) = 2$ or
  (ii) $ub = \lrsuf(Tc)$ and $\#\occ_{Tc}(ub) = 2$ hold.
\end{lemma}
\begin{proof}
  Since $aub \in \NUS(T)$, $au$ is unique in $T$, $ub$ is unique in $T$, and $u$ is repeating in $T$.
  Also, since $aub \not\in \NUS(Tc)$, $au$ or $ub$ becomes repeating in $Tc$.
  See also Figure~\ref{fig:deletedNOcc}.
  Note that $au$ and $ub$ cannot be repeating in $Tc$, because if we assume they become repeating in $Tc$ simultaneously,
  then both $au$ and $ub$ are suffixes of $Tc$ and this implies $au = ub$, which contradicts that $au$ is unique in $T$.
  \begin{figure}[tbp]
    \centering
    \includegraphics[width=0.8\linewidth]{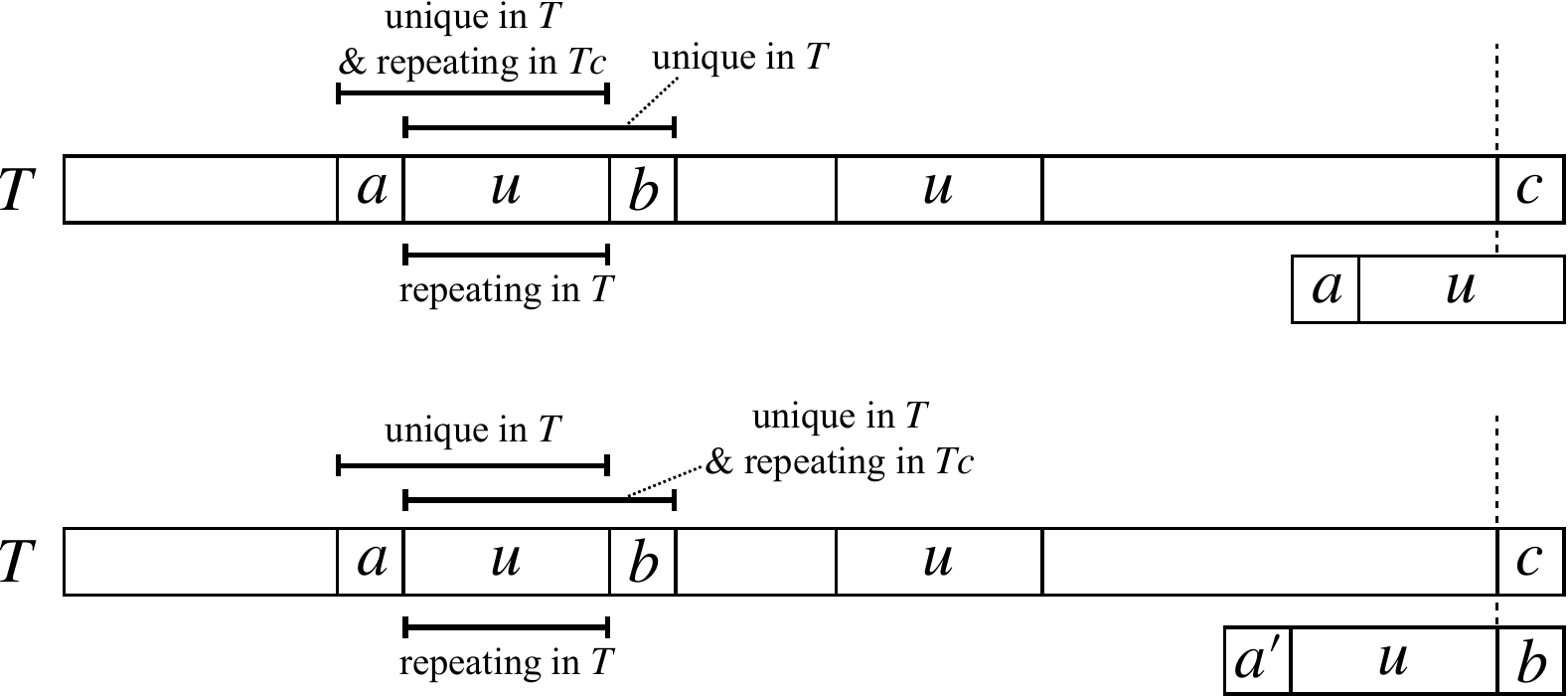}
    \caption{Illustration for Lemma~\ref{lem:deletedNOcc}. The upper part depicts the case where $au$ is repeating in $Tc$. The lower part depicts the case where $ub$ is repeating in $Tc$}
    \label{fig:deletedNOcc}
  \end{figure}

  If $au$ is repeating in $Tc$, then $au$ is a suffix of $Tc$.
  Since $au$ is unique in $T$, $\#\occ_{Tc}(au) = 2$ holds.
  Also, since $u$ is repeating in $T$, $\#\occ_{Tc}(u) \ge 3$ holds.
  These imply that $au = \sqsuf(Tc)$.

  If $ub$ is repeating in $Tc$, then $ub$ is a suffix of $Tc$.
  Since $ub$ is unique in $T$, $\#\occ_{Tc}(ub) = 2$ holds.
  Now let $a' = T[|T|-|u|]$ be the preceding character of the suffix $ub$ of $Tc$.
  If we assume that $a'ub$ is repeating in $Tc$,
  then the other occurrence $a'ub$ matches the occurrence of $aub$ since $ub$ is unique in $T$.
  This implies that $a' = a$, which contradicts $au$ is unique in $T$.
  Thus, $ub = \lrsuf(Tc)$.
\end{proof}

\begin{lemma}\label{lem:addedNOccNonsuffix}
  Suppose that $aub \in \NUS(Tc)\setminus\NUS(T)$ be a net unique substring in $Tc$
  such that $aub$ is not a suffix of $Tc$
  where $a, b, c \in \Sigma$ and $u \in \Sigma^\star$.
  Then, $u = \lrsuf(Tc)$ and $\#\occ_{Tc}(u) = 2$ hold.
\end{lemma}
\begin{proof}
  Since $aub \in \NUS(Tc)$, $au$ is unique in $Tc$, $ub$ is unique in $Tc$, and $u$ is repeating in $Tc$.
  Also, since $aub \not\in \NUS(T)$, $u$ is unique in $T$.
  Namely, $u$ occurs in $Tc$ as a suffix and $\#\occ_{Tc}(u) = 2$.
  See also Figure~\ref{fig:addedNOccNonsuffix}.
  Since $au$ is unique in $Tc$, the preceding character $a'$ of $u = T[|Tc|-|u|+1.. |Tc|]$ is not equal to $a$,
  and thus, $a'u$ is unique in $Tc$.
  Therefore, $u = \lrsuf(Tc)$ holds.
  \begin{figure}[tbp]
    \centering
    \includegraphics[width=0.8\linewidth]{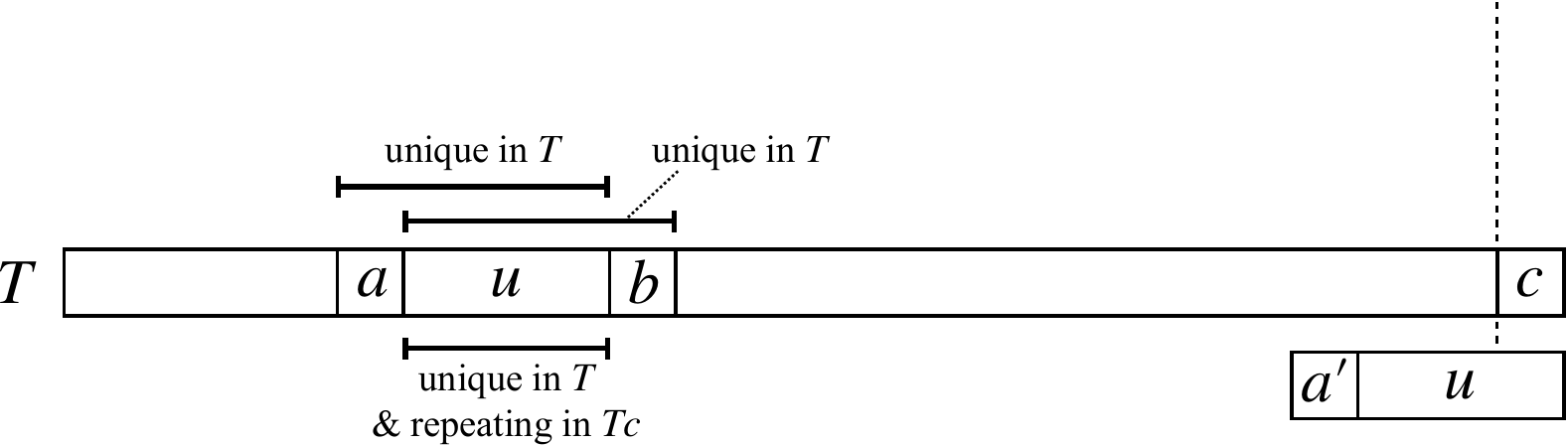}
    \caption{Illustration for Lemma~\ref{lem:addedNOccNonsuffix}. Suffix $u$ of $Tc$ is repeating in $Tc$ but longer suffix $a'u$ of $Tc$ is unique in $Tc$.}
    \label{fig:addedNOccNonsuffix}
  \end{figure}
\end{proof}

\begin{lemma}\label{lem:addedNOccsuffix}
  Suppose that $auc \in \NUS(Tc)\setminus\NUS(T)$ be a net unique substring in $Tc$
  such that $auc$ is a suffix of $Tc$ 
  where $a, c \in \Sigma$ and $u \in \Sigma^\star$.
  Then, either (i) $u = \lrsuf(T)$ or (ii) $u = c^{\mathit{exp}}$ and $\#\occ_T(u) = 1$ where $\mathit{exp} = \max\{e\mid c^e \text{~is a suffix of~} T\}$.
\end{lemma}
\begin{proof}
  Since $auc \in \NUS(Tc)$, $au$ is unique in $Tc$, $uc$ is unique in $Tc$, and $u$ is repeating in $Tc$.
  There are two cases with respect to the number of occurrences of $u$ in $T$.
  See also Figure~\ref{fig:addedNOccsuffix}.
  \begin{figure}[tbp]
    \centering
    \includegraphics[width=0.8\linewidth]{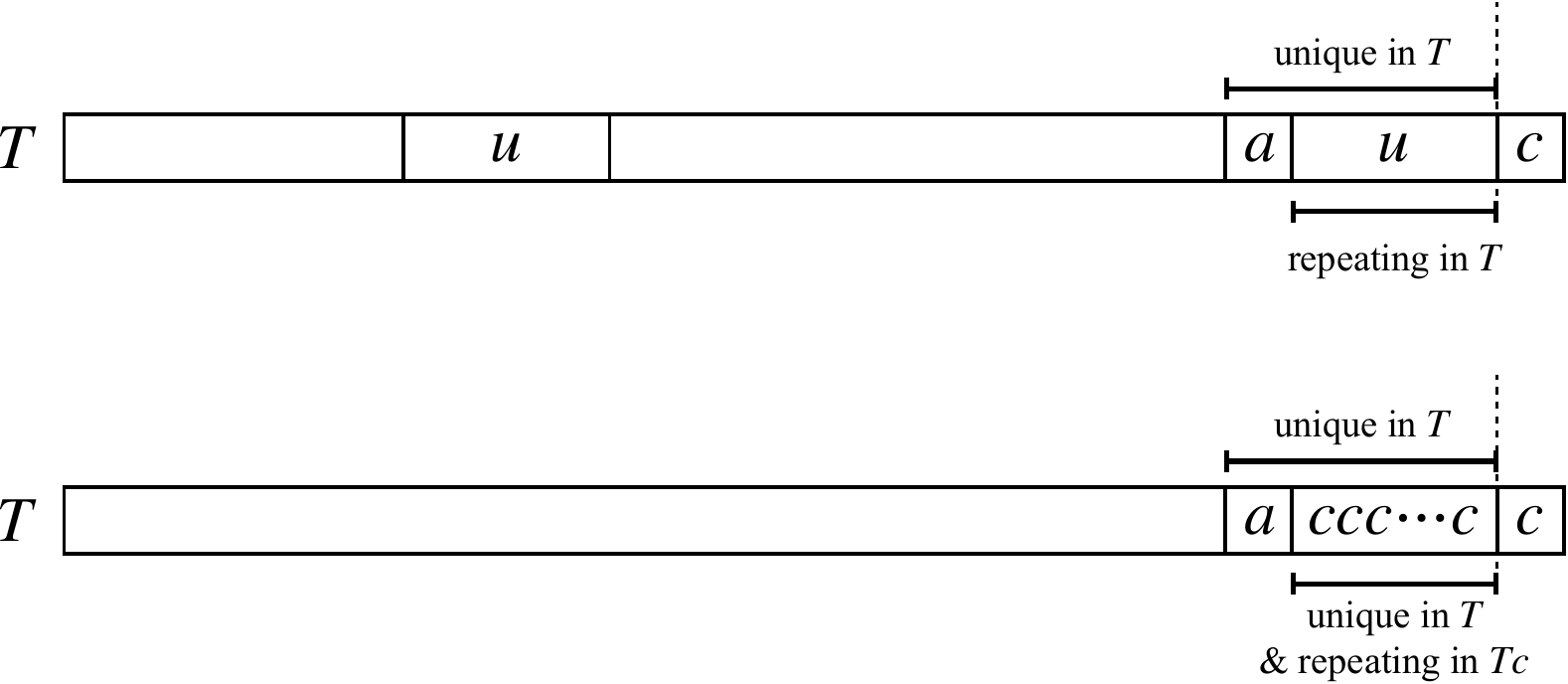}
    \caption{Illustration for Lemma~\ref{lem:addedNOccsuffix}. The upper part depicts the case where $u$ is repeating in $T$. The lower part depicts the case where $u$ is unique in $T$.}
    \label{fig:addedNOccsuffix}
  \end{figure}

  If $u$ is repeating in $T$, $u = \lrsuf(T)$ since $au$ is unique in $T$.
  If $u$ is unique in $T$, then appending $c$ causes $u$ to be repeating in $T$.
  This implies that $u = u[2..|u|]c$, i.e., $u = c^{|u|}$.
  Also, since $uc = c^{|u|+1}$ is unique in $Tc$, $a \ne c$ holds.
  Thus $|u| = \mathit{exp} = \max\{e\mid c^e \text{~is a suffix of~} T\}$.
\end{proof}

As for the differences between $\NUS(T)$ and $\NUS(cT)$,
the three following lemmas also hold by symmetry:
\begin{lemma}\label{lem:sym_deletedNOcc}
  Suppose that $aub \in \NUS(T)\setminus\NUS(cT)$ be a net unique substring in $T$
  where $a, b, c \in \Sigma$ and $u \in \Sigma^\star$.
  Then,
  (i) $au = \sqpref(cT)$ and $\#\occ_{cT}(au) = 2$ or
  (ii) $ub = \lrpref(cT)$ and $\#\occ_{cT}(ub) = 2$ hold.
\end{lemma}

\begin{lemma}\label{lem:sym_addedNOccNonprefix}
  Suppose that $aub \in \NUS(cT)\setminus\NUS(T)$ be a net unique substring in $cT$
  such that $aub$ is not a prefix of $cT$
  where $a, b, c \in \Sigma$ and $u \in \Sigma^\star$.
  Then, $u = \lrpref(cT)$ and $\#\occ_{cT}(u) = 2$ hold.
\end{lemma}

\begin{lemma}\label{lem:sym_addedNOccprefix}
  Suppose that $auc \in \NUS(cT)\setminus\NUS(T)$ be a net unique substring in $cT$
  such that $auc$ is a prefix of $cT$ 
  where $a, c \in \Sigma$ and $u \in \Sigma^\star$.
  Then, either (i) $u = \lrpref(T)$ or (ii) $u = c^{\mathit{exp}}$ and $\#\occ_T(u) = 1$ where $\mathit{exp} = \max\{e\mid c^e \text{~is a prefix of~} T\}$.
\end{lemma}
\section{Algorithms}\label{sec:algo}

In this section, we present our online/sliding algorithms for computing
extended net occurrences of repeats for a given string.

\subsection{Online algorithm based on implicit suffix trees}

By Lemmas~\ref{lem:deletedNOcc},~\ref{lem:addedNOccNonsuffix} and~\ref{lem:addedNOccsuffix},
we can compute $\ENO(Tc)$ from $\ENO(T)$ with Algorithm~\ref{alg:online}
since $[p.. q] \in \ENO(T)$ iff $T[p.. q] \in \NUS(T)$.
We encode each element $[i..j] \in \ENO(T)$ by a pair $(i,j)$ so that $\ENO(T)$ can be stored in $O(\#\ENO(T))$ space.
Note that $\#\ENO(Tc) - \#\ENO(T) \ne -2$ while $\#(\ENO(Tc)\setminus\ENO(T))$ can be $2$.
See lines~9--15 of Algorithm~\ref{alg:online}.
The size of $E$ decreases by $1$ if we enter line~12, however, the size increases by $1$ at line~14.
Thus $-1 \le \#\ENO(Tc) - \#\ENO(T) \le 2$.
\begin{algorithm}[th]
  \caption{Updating extended net occurrences when a character is appended}\label{alg:online}
  \begin{algorithmic}[1]
    \Require ${E} = \ENO(T)$
    \Ensure  ${E} = \ENO(Tc)$
    \Procedure{AppendChar}{String $T$, character $c$, set ${E}$}
    \If{$\#\occ_{Tc}(\sqsuf(Tc)) = 2$} \label{line:2}
    \State $(i, j) \leftarrow$ the non-suffix occurrence of $\sqsuf(Tc)$ in $Tc$. \label{line:3}
    \If{$(i, j+1) \in {E}$} \label{line:4}
    \State ${E} \leftarrow {E} \setminus \{(i, j+1)\}$ \Comment{by (i) of Lemma~\ref{lem:deletedNOcc}} \label{line:5}
    \EndIf
    \EndIf
    \If{$\#\occ_{Tc}(\lrsuf(Tc)) = 2$} \label{line:8}
    \State $(i', j') \leftarrow$ the non-suffix occurrence of $\lrsuf(Tc)$ in $Tc$. \label{line:9}
    \If{$i' > 1$}
    \If{$(i'-1, j') \in {E}$} \label{line:11}
    \State ${E} \leftarrow {E} \setminus \{(i'-1, j')\}$ \Comment{by (ii) of Lemma~\ref{lem:deletedNOcc}} \label{line:12}
    \EndIf
    \State ${E} \leftarrow {E} \cup \{(i'-1, j'+1)\}$ \Comment{by Lemma~\ref{lem:addedNOccNonsuffix}} \label{line:14}
    \EndIf
    \EndIf
    \If{$|\lrsuf(Tc)| \le |\lrsuf(T)|$} \Comment{$\iff \#\occ_{Tc}(\lrsuf(T)c) = 1$} \label{line:17}
    \State ${E} \leftarrow {E} \cup \{(|T|-|\lrsuf(T)|, |Tc|)\}$ \Comment{by (i) of Lemma~\ref{lem:addedNOccsuffix}} \label{line:18}
    \ElsIf{$\#\occ_{T}(c^{\mathit{exp}}) = 1$ where $\mathit{exp} = \max\{e\mid c^e \text{~is a suffix of~} T\}$} \label{line:19}
    \State ${E} \leftarrow {E} \cup \{(|T|-\mathit{exp}, |Tc|)\}$ \Comment{by (ii) of Lemma~\ref{lem:addedNOccsuffix}} \label{line:20}
    \EndIf
    \EndProcedure
  \end{algorithmic}
\end{algorithm}

From Algorithm~\ref{alg:online}, we obtain the following theorem:
\begin{theorem}\label{thm:main}
  Given string $T$, $E = \ENO(T)$, and character $c$,
  we can compute $\ENO(Tc)$ in $O(t(n))$ time using $O(s(n))$ space
  if each of the following operations can be executed in $t(n)$ time within $s(n)$ space:
  \begin{enumerate}
    \item Determine if $\#\occ_{Tc}(\sqsuf(Tc)) = 2$.
    \item Find the non-suffix occurrence of $\sqsuf(Tc)$ in $Tc$ when $\#\occ_{Tc}(\sqsuf(Tc)) = 2$.
    \item Determine if $\#\occ_{Tc}(\lrsuf(Tc)) = 2$.
    \item Find the non-suffix occurrence of $\lrsuf(Tc)$ in $Tc$ when $\#\occ_{Tc}(\lrsuf(Tc)) = 2$.
    \item Compute the lengths of $\lrsuf(T)$ and $\lrsuf(Tc)$.
    \item Compute $\mathit{exp} = \max\{e\mid c^e \text{~is a suffix of~} T\}$ and determine if $\#\occ_{T}(c^{\mathit{exp}}) = 1$.
    \item Determine if $(i, j) \in E$ for given pair $(i, j)$.
    \item Insert an element $(i, j)$ into $E$ for given pair $(i, j)$.
    \item Delete an element $(i, j)$ from $E$ if $(i, j) \in E$ for given pair $(i, j)$.
  \end{enumerate}
\end{theorem}
\begin{proof}
  Look at Algorithm~\ref{alg:online}.
  Line~\ref{line:2} uses operation 1,
  Line~\ref{line:3} uses operation 2,
  Lines~\ref{line:4} and~\ref{line:11} use operation 7,
  Lines~\ref{line:5} and~\ref{line:12} use operation 9,
  Line~\ref{line:8} uses operation 3, 
  Line~\ref{line:9} uses operation 4, 
  Lines~\ref{line:14}, \ref{line:18}, and~\ref{line:20} use operation 8, 
  Line~\ref{line:17} uses operation 5, and
  Line~\ref{line:19} uses operation 6.
  Thus, Algorithm~\ref{alg:online} consists of the above nine operations and basic arithmetic operations.
\end{proof}
Note that the correctness of Theorem~\ref{thm:main} does not depend on the data structure used.
The next lemma holds if we utilize the \emph{implicit suffix tree} by Ukkonen~\cite{Ukkonen95}.

\begin{lemma}\label{lem:Ukkonen}
  Based on Ukkonen's left-to-right online suffix tree construction~\cite{Ukkonen95}, we can design a data structure $\mathcal{D}_T$ of size $O(|T|)$
  that supports all nine operations of Theorem~\ref{thm:main} in constant time.
  The data structure $\mathcal{D}_T$ can be updated to $\mathcal{D}_{Tc}$
  in amortized $O(\log\sigma)$ time where $c$ is a character.
\end{lemma}
\begin{proof}
  We employ an implicit suffix tree~\cite{Ukkonen95} of online string $T$
  enhanced with the \emph{active point}, which represents $\lrsuf(T)$, and the \emph{secondary active point}, which represents $\sqsuf(T)$ as in~\cite{MienoFNIBT22}.
  According to~\cite{MienoFNIBT22}, such an enhanced implicit suffix tree can support operations 1--6 in $O(1)$ time.
  For the readers of this paper, we briefly explain how to perform those operations efficiently below:
  \begin{itemize}
    \item Operation 5 is obvious since we maintain the active point for every step.
    \item Operation 3 can be easily done in constant time by checking whether the active point locates on an edge towards a leaf or not.
    \item Operation 1 can be done in constant time by using the (secondary) active points (due to Lemma~1 of~\cite{MienoFNIBT22}).
    \item Operations 2 and 4 can be done by looking at the leaves under the (secondary) active points.
      For instance, look at the implicit suffix tree $\STree'(T)$ depicted in Figure~\ref{fig:stree}.
      The secondary active point (the gray star), which represents the shortest quasi-unique suffix $s = \mathtt{abbab}$, is on an edge towards a leaf.
      Further, the leaf under the secondary active point represents the suffix $\mathtt{abbabbab}$ of $T$ starting at position $5$.
      Thus, $s$ occurs at position $5$, which is the non-suffix occurrence of $s$.
      Similarly, the non-suffix occurrence of the longest repeating suffix $\mathtt{bbabbab}$ is position $3$
      since the leaf under the active point (the white star) represents the suffix of $T$ starting at position $3$. 
    \item Value $\mathit{exp}$ defined in operation 6 can be easily maintained independent of the suffix tree.
  \end{itemize}
  As for operations 7--9,
  we implement set $E = \ENO(T)$ as a set of occurrences
  where each element $(i, j) \in E$ is connected to the corresponding locus of the suffix tree.
  Since $T[i.. j]$ is unique in $T$, the locus of $T[i.. j]$ is
  either the leaf corresponding to unique suffix $T[i..|T|]$ or on the edge towards the leaf.
  Thus we can perform operations 7--9 in constant time via the leaves of the suffix tree,
  for given $(i, j)$, which represents some unique substring of $T$.

  Finally, the data structure can be maintained in amortized $O(\log\sigma)$ time:
  Basically the amortized analysis is due to~\cite{Ukkonen95}.  
  The secondary active point, which was originally proposed in~\cite{MienoFNIBT22},
  can also be maintained in a similar manner to the active point, and thus
  the amortized analysis for the secondary active point is almost the same as that for the active point in~\cite{Ukkonen95}~(see~\cite{MienoFNIBT22} for the complete proof).
\end{proof}

By wrapping up the above discussions, we obtain the following theorem:
\begin{theorem}
  We can compute the set of extended net occurrences of string $T$ of length $n$
  given in an online manner in a total of $O(n \log\sigma)$ time using $O(n)$ space.
\end{theorem}

\subsection{Sliding-window algorithm based on implicit suffix trees}

By applying symmetric arguments of Theorem~\ref{thm:main}, we can design a sliding-window algorithm.
\begin{lemma}
  There exists a data structure $\mathcal{D}_T$ of size $O(|T|)$
  that supports all nine operations of Theorem~\ref{thm:main}
  in addition to their symmetric operations listed below in constant time.
  \begin{enumerate}
    \item Determine if $\#\occ_{T}(\sqpref(T)) = 2$.
    \item Find the non-prefix occurrence of $\sqpref(T)$ in $T$ when $\#\occ_{T}(\sqpref(T)) = 2$.
    \item Determine if $\#\occ_{T}(\lrpref(T)) = 2$.
    \item Find the non-prefix occurrence of $\lrpref(T)$ in $T$ when $\#\occ_{T}(\lrpref(T)) = 2$.
    \item Compute the lengths of $\lrpref(T[2..n])$ and $\lrpref(T)$.
    \item Compute $\mathit{exp'} = \max\{e\mid c^e \text{~is a prefix of~} T[2..n]\}$ and determine if $\#\occ_{T[2..n]}(c^{\mathit{exp'}}) = 1$ where $c = T[1]$.
  \end{enumerate}
  The data structure $\mathcal{D}_T$ can be updated to either
  $\mathcal{D}_{T[2..|T|]}$ or $\mathcal{D}_{Tc}$ in amortized $O(\log\sigma)$ time
  where $c$ is a character.
\end{lemma}
\begin{proof}
  The sliding suffix tree data structure of~\cite{MienoFNIBT22} supports
  all the operations in amortized $O(\log\sigma)$ time using $O(d)$ space.
\end{proof}

In case we perform the deletion of the leftmost character and 
the addition of the rightmost character simultaneously,
then our algorithm works for a sliding-window of fixed size $d$.
On the other hand, our scheme is also applicable to a sliding-window of variable size.
Thus we have the following:
\begin{theorem}\label{theo:slidingENO}
  We can maintain the set of extended net occurrences for a sliding window over string $T$ of length $n$
  in a total of $O(n \log\sigma)$ time using $O(d)$ working space
  where $d$ is the maximum size of the window.
\end{theorem}

\subsection{Online algorithm based on implicit CDAWGs}

The next lemma is an adaptation of Lemma~\ref{lem:Ukkonen} which uses implicit CDAWGs in place of implicit suffix trees:
\begin{lemma}\label{lem:CDAWG}
  Based on the left-to-right online CDAWG construction~\cite{InenagaHSTAMP05}, 
  we can design a data structure $\mathcal{C}_T$ of size $O(\e(T))$
  that supports all nine operations of Theorem~\ref{thm:main} in constant time.
  The data structure $\mathcal{C}_T$ can be updated to $\mathcal{C}_{Tc}$
  in amortized $O(\log\sigma)$ time where $c$ is a character.
\end{lemma}

\begin{proof}
  Since the online implicit CDAWG construction algorithm~\cite{InenagaHSTAMP05} is based on Ukkonen's 
  implicit suffix tree construction,
  it also maintains the active point that indicates the locus corresponding to $\lrsuf(T)$.
  While the locus can correspond to multiple substrings of $T$ (as the CDAWG is a DAG),
  we can retrieve $|\lrsuf(T)|$ in $O(1)$ time by storing, in each node $v$ of $\CDAWG'(T)$,
  the length of the maximal repeat corresponding to $v$.
  This is because the path that spells out $\lrsuf(T)$ from the source consists only of the 
  primary edges (see~\cite{InenagaHSTAMP05} for more details).
  Since edge label $x$ is represented by an integer pair $(p,q)$ such that $x = T[p..q]$,
  we can obtain the non-suffix occurrence $(i',j')$ of $\lrsuf(T)$ in $O(1)$ time (Line~\ref{line:9} in Algorithm~\ref{alg:online}).

  The secondary active point that indicates the locus for $\sqsuf(T)$ can also be maintained 
  on the implicit $\CDAWG'(T)$ by adapting the algorithm from~\cite{MienoFNIBT22}.
  Let $y$ be the suffix of $T$ that is one-character shorter than $\sqsuf(T)$.
  By definition, $y$ is the longest suffix of $T$ such that $\#\occ_T(y) \geq 3$.
  Given the locus $P$ for $y$ on $\CDAWG'(T)$, one can check in $O(1)$ time whether the substrings corresponding to $P$
  occur at least 3 times, by checking the number of paths from $P$ to the sink
  and checking if the active point is in the subgraph under $P$.
  Also, by definition, $y$ is the longest string represented by the locus $P$.
  This tells us the length of $\sqsuf(T)$ as well.
  Thus, we can also maintain the secondary active point in a similar manner to 
  the active point on the implicit CDAWG in $O(\log \sigma)$ amortized time per character,
  and we can obtain the non-suffix occurrence $(i,j)$ of $\sqsuf(T)$ in $O(1)$ time (Line~\ref{line:3} in Algorithm~\ref{alg:online}).

  The $O(\e(T))$-space requirement follows from Lemma~\ref{lem:implicit_CDAWG_size}.
\end{proof}

\begin{theorem} \label{theo:CDAWG_ENO}
  We can compute the set of extended net occurrences of string $T$ of length $n$
  given in an online manner in a total of $O(n \log\sigma)$ time using $O(\e(T))$ working space.
\end{theorem}

\begin{proof}
  The correctness and the time complexity follows from the above discussions.

  It is shown in~\cite{InenagaHSTAMP05} that 
  the function $\e'(T[1..i])$ is monotonically non-decreasing for any online string $T[1..i]$ 
  with increasing $i = 1, \ldots, n$.
  Together with Lemma~\ref{lem:implicit_CDAWG_size}, 
  we have $\e'(T[1..i]) \leq \e'(T) \leq \e(T)$ for any $1 \leq i \leq n$,
  which leads to an $O(\e(T))$-space bound.
\end{proof}
\section{Relating extended net occurrences and MUSs}

In this section, we give a full characterization of the extended net occurrences of repeats in $T$
in terms of minimal unique substrings (MUSs) in $T$.
See also Figure~\ref{fig:MUS_ENO} for illustration.

\begin{figure}[htb]
  \centering
  \includegraphics[scale=0.55]{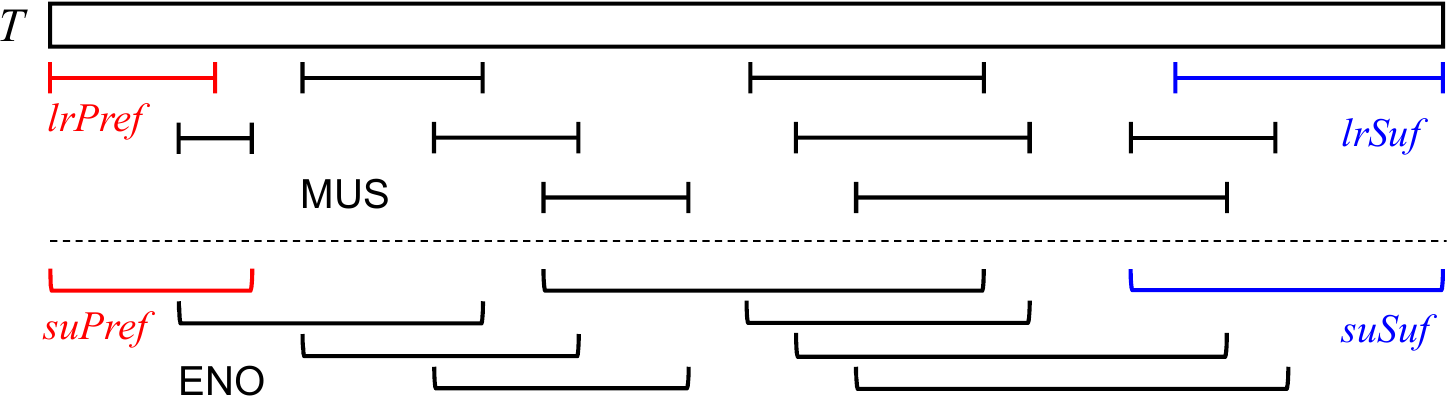}
  \caption{Illustration for Lemma~\ref{lem:ENO_MUS} and Lemma~\ref{lem:MUS_ENO-new}.}
  \label{fig:MUS_ENO}
\end{figure}

\begin{lemma}\label{lem:ENO_MUS}
  Let $[i-1..h], [k..j+1] \in \MUS(T)$ be the intervals that represent
  consecutive MUSs in $T$, namely, there is no element $[s..t]$ in $\MUS(T)$ such that $i \leq s < k$ and $h < t \leq j$.
  Then, $[i..j]$ is a net occurrence for repeat $u = T[i..j]$.
\end{lemma}

\begin{proof}
  Observe that any unique substring of $T$ must contain a MUS in $T$.
  Since $T[i..j]$ does not contain a MUS, $u = T[i..j]$ is a repeat in $T$.
  Let $a = T[i-1]$ and $b = T[j+1]$.
  Then, $\#\occ_T(aub) = \#\occ_T(au) = \#\occ_T(ub) = 1$ since any of $aub = T[i-1..j-1]$, $au = T[i-1..j]$, and $ub = T[i..j+1]$ contains a MUS.
\end{proof}
A consequence of Lemma~\ref{lem:ENO_MUS} is that
a net occurrence $[i..j]$ cannot be contained in another net occurrence $[i'..j']$.
This is because a MUS cannot be contained in another MUS.
Another consequence is that two consecutive extended net occurrences in $\ENO(T)$ are overlapping.

Below we show the reversed version of Lemma~\ref{lem:ENO_MUS}:

\begin{lemma}\label{lem:MUS_ENO-new}
  Let $\mathsf{L} = \ENO(T) \cup \{[1..p] \mid p = |\lrpref(T)|+1\} \cup \{[q..n] \mid q = n - |\lrsuf(T)|\}$.
  Let $[h..j], [i..k] \in \mathsf{L}$ be consecutive elements in $\mathsf{L}$, namely, 
  there is no element $[s..t]$ in $\mathsf{L}$ such that $h < s < i$ and $j < t < k$.
  Then, $[i..j] \in \MUS(T)$.
\end{lemma}

\begin{proof}
  By the definition of the extended net occurrences,
  there is a MUS $[x..j]$ with $x \ge h+1$ that ends at $j$ since $T[h+1.. j]$ is unique and $T[h+1.. j-1]$ is repeating in $T$.
  Similarly, there is a MUS $[i..y]$ with $y \le k-1$ that starts at $i$.
  Here, for the sake of contradiction, we assume $i \ne x$.
  If $i < x$, there are at least two MUSs within range $[i, j] \subset [h, k]$.
  If $i > x$, there are at least two MUSs within range $[x, y] \subset [h, k]$.
  In both cases, there exists some net occurrence within range $[h, k]$ by Lemma~\ref{lem:ENO_MUS}, which contradicts that $[h..j]$ and $[i..k]$ are consecutive elements in $\mathsf{L}$.
  Thus $i = x$ holds.
  Similarly, we can prove $j = y$, hence $[i..j] \in \MUS(T)$.

  Consider the case where $h = 1$ and $j = p$, namely $T[1..p]$ is the \emph{shortest unique prefix} (\emph{suPref}) of $T$,
  and $T[i..k]$ is the leftmost extended net occurrence in $T$.
  Again by the definition of the extended net occurrences, there is a MUS that begins at position $i$.
  Let $T[1..p] = ub$ where $u \in \Sigma^\star$ and $b \in \Sigma$.
  Then, $u = \lrpref(T)$.
  Since $ub$ is unique and since $u = \lrpref(T)$,
  there must exist a MUS that ends at position $p$ (see Figure~\ref{fig:MUS_ENO}.)
  Using a similar argument as above, it can be proven that these two MUSs are the same.
  Thus $[i..p] \in \MUS(T)$.
  The case where $i = q$ and $k = n$ is symmetric.
\end{proof}

Consequently, the next theorem follows from Lemma~\ref{lem:ENO_MUS} and Lemma~\ref{lem:MUS_ENO-new}.
\begin{theorem} \label{theo:ENO_MUS}
  For any string $T$, 
  \begin{enumerate}
    \item[(1)] $\#\ENO(T) = \#\MUS(T)-1$.
    \item[(2)] $\ENO(T)$ can be obtained from the sorted $\MUS(T)$ in optimal $O(\#\ENO(T))$ time.
    \item[(3)] $\MUS(T)$ can be obtained from the sorted $\ENO(T)$, $|\lrpref(T)|$, and $|\lrsuf(T)|$ in optimal $O(\#\ENO(T))$ time.
  \end{enumerate}
\end{theorem}

We also have the following corollary for space-efficient computation of MUSs:
\begin{corollary} \label{coro:CDAWG_MUS}
  We can maintain the set of all MUSs of a string $T$ of length $n$
  given in an online manner in a total of $O(n \log\sigma)$ time using $O(\e(T))$ working space,
  where $\e(T)$ denotes the size of $\CDAWG(T)$.
\end{corollary}

\begin{proof}
  By combining Theorem~\ref{theo:CDAWG_ENO} and Lemmas~\ref{lem:ENO_MUS} and~\ref{lem:MUS_ENO-new},
  we obtain the corollary except for computation of $|\lrpref(T)|$.
  This can easily be maintained in the implicit CDAWG as follows.
  Let $z$ be the node of $\CDAWG'(T)$ from which the primary edge to the sink stems out.
  We identify $z$ with the maximal repeat that the node represents.
  If the active point does not exist on this primary edge, then $|z| = |\lrpref(T)|$.
  If the active point lies on this primary edge leading to the sink,
  then $|z|+k = |\lrpref(T)|$, where $k$ is the offset of the active point on the primary edge from the node $z$.
\end{proof}

For any string $T$, $\#\MUS(T) \leq \e(T)$ holds~\cite{InenagaMAFF24}.
Together with Theorem~\ref{theo:ENO_MUS}, we obtain:
\begin{corollary}
  For any string $T$, $\#\ENO(T) < \e(T)$ holds.
\end{corollary}

\section{Conclusions and open questions}

In this paper we presented how Ukkonen's left-to-right online suffix tree construction 
can be used for online computation of string net frequency.
Our main contributions are space-efficient algorithms for computing string net occurrences,
one works in $O(d)$ space in the sliding model for window-length $d$,
and the other works in $O(\e(T))$ space where $\e(T)$ denotes the size of the CDAWG of the input string $T$.
Both of our methods run in $O(n \log \sigma)$ time 
and can report all (extended) net occurrences of repeats in the current string in output-optimal time.
We also showed that computing the sorted list of extended net occurrences of repeats in a string $T$
is equivalent to computing the sorted list of minimal unique substrings (MUSs) in $T$.

An intriguing open question is whether one can efficiently compute the
extended net occurrences of repeats within $O(\rb)$ space, where $\rb$ denotes the number of 
equal-character runs in the BWT of the input string.
It is known that $\rb \leq \e$ holds for any string~\cite{BelazzouguiCGPR15}.
The \emph{R-enum} algorithm of Nishimoto and Tabei~\cite{NishimotoT21} is able to compute
the set of MUSs in $O(n \log \log_\omega (n/\rb))$ time with $O(\rb)$ space,
where $\omega$ denotes the machine word size of the word RAM model.
However, it is unclear whether their algorithm can output a list of MUSs
arranged in the sorted order of the beginning positions within $O(\rb)$ space.

\end{document}